\documentclass[11pt,a4paper,onecolumn]{quantumarticle}
\pdfoutput=1 
\usepackage[T1]{fontenc}

\usepackage{tikz}
\usetikzlibrary{positioning}
\usetikzlibrary{calc}
\usetikzlibrary{arrows}
\usepackage{mdwlist}
\usepackage{thmtools}
\usepackage{mathtools}
\usepackage{physics}
\usepackage{array}
\usepackage[numbers, sort&compress]{natbib}
\usetikzlibrary{decorations.pathmorphing}
\usepackage{mathrsfs}

\tikzset{snake it/.style={decorate, decoration=snake}}

\usetikzlibrary{decorations.pathreplacing,decorations.markings}

\usepackage{hyperref}
\usepackage[capitalize,nameinlink]{cleveref}
\usepackage{float}		% this is to place figures where requested!
\usepackage{graphicx}		% needed for the figures
\usepackage{subcaption}

\DeclareMathOperator{\SMP}{\mathrm{SMP}}
\DeclareMathOperator{\opt}{\text{opt}}

\newcommand{\bc}{\cos^2(\pi/8)}

\newcommand{\IPfunc}{\mathsf{IP}}

\tikzset{
    >=stealth',
    punkt/.style={
           rectangle,
           rounded corners,
           draw=black, very thick,
           text width=6.5em,
           minimum height=2em,
           text centered},
    pil/.style={
           ->,
           thick,
           shorten <=2pt,
           shorten >=2pt,},
  on each segment/.style={
    decorate,
    decoration={
      show path construction,
      moveto code={},
      lineto code={
        \path [#1]
        (\tikzinputsegmentfirst) -- (\tikzinputsegmentlast);
      },
      curveto code={
        \path [#1] (\tikzinputsegmentfirst)
        .. controls
        (\tikzinputsegmentsupporta) and (\tikzinputsegmentsupportb)
        ..
        (\tikzinputsegmentlast);
      },
      closepath code={
        \path [#1]
        (\tikzinputsegmentfirst) -- (\tikzinputsegmentlast);
      },
    },
  },
  mid arrow/.style={postaction={decorate,decoration={
        markings,
        mark=at position .5 with {\arrow[#1]{stealth'}}
      }}}
}

\newtheorem{theorem}{Theorem}

\newtheorem{corollary}[theorem]{Corollary}

\newtheorem{definition}[theorem]{Definition}

\newtheorem{lemma}[theorem]{Lemma}

\newtheorem{proposition}[theorem]{Proposition}
\newtheorem{remark}[theorem]{Remark}

\newenvironment{proof}[1][Proof]{\noindent\textbf{#1.}}{\ \rule{0.5em}{0.5em}}

\begin{document} 

\title{Quantum gate lower bounds for loss-tolerant position verification}

\author[1,2]{Alex May}
\email{amay@perimeterinstitute.ca}
\orcid{0000-0002-4030-5410}

\author[3]{Philip Verduyn Lunel}
\email{Philip.Verduyn-Lunel@lip6.fr}
\orcid{0000-0001-5419-6027}

\affiliation[1]{Institute for Quantum Computing, Waterloo, Ontario}
\affiliation[2]{Perimeter Institute for Theoretical Physics, Waterloo, Ontario}
\affiliation[3]{Sorbonne Universit\'e, CNRS, LIP6, France}

\begin{abstract}
Quantum position-verification is a cryptographic task wherein a verifier attempts to establish the location in space of a prover.  
Recent experiments have implemented a well-studied class of position-verification schemes, known as the $f$-BB84 scenario, but their security under realistic loss and imperfect state preparation remains incompletely understood. 
We give a new lower bound on any attack on this scheme, in particular proving nearly-linear quantum gate lower bounds on the attacker, even when allowing the attacker to declare a transmission loss of up to $50\%$, allowing for the challenges prepared by the referee to be imperfect, and allowing the quantum messages used in the protocol to be arbitrarily slow.
Our results are applicable to recent and upcoming experimental implementations of $f$-BB84, and in particular establish their security under a bounded quantum gate assumption on the attacker.
The key ingredient is a tight analytic tradeoff for a lossy BB84 monogamy-of-entanglement game, valid without requiring the attackers’ outputs to agree, which replaces observations previously made numerically.
\end{abstract}

\vfill

\maketitle

\pagebreak

\tableofcontents

%%%%%%%%%%%%%%%%%%%%%%%%%%%%%%%%%%%%%%%%%%%%%%%%%%%%%%%%
\section{Introduction}
%%%%%%%%%%%%%%%%%%%%%%%%%%%%%%%%%%%%%%%%%%%%%%%%%%%%%%%%

Quantum position-verification (QPV) \cite{kent2006tagging, buhrman2014position, kent2011quantum} is a proposed method of verifying the location of a party or device. 
The technique relies on timing constraints coming from relativity, along with restrictions inherent to quantum information. 
See \cref{fig:2dsetup} for a typical set-up. 
QPV is of interest as a practical cryptographic tool, and for its connections to a diverse set of other subjects in cryptography, complexity theory, and physics; see \cite{may2026entanglement} for a review.
A number of recent experiments have implemented position-verification schemes \cite{cowperthwaite2023towards, kanneworff2025towards, kavuri2025device,fan2026relativistic}.
However, the level of security achieved in these or near-term experiments is not fully understood.
Here we prove new bounds on any attacker in an experimentally feasible class of QPV schemes known as $f$-BB84 schemes. 

\begin{figure}
    \centering
    \begin{subfigure}{0.45\textwidth}
    \begin{tikzpicture}[scale=0.6]
    
    \node[below left] at (-4,0) {$c_1$};
    \draw[fill=black] (-4,0) circle (0.15);

    \node[below right] at (4,0) {$c_2$};
    \draw[fill=black] (4,0) circle (0.15);

    \node[below right] at (4,8) {$r_2$};
    \draw[fill=blue] (4,8) circle (0.15);

    \node[below left] at (-4,8) {$r_1$};
    \draw[fill=blue] (-4,8) circle (0.15);
    
    \draw[fill=gray,opacity=0.5] (-1,1) -- (1,1) -- (1,7) -- (-1,7) -- (-1,1);

    \draw[->] (-5.5,1) -- (-5.5,2);
    \node[above] at (-5.5,2) {$t$};
    \draw[->] (-5.5,1) -- (-4.5,1);
    \node[right] at (-4.5,1) {$x$};
    
    \end{tikzpicture}
    \caption{}
    \label{fig:taggingsub1}
    \end{subfigure}
    \hfill
\begin{subfigure}{.45\textwidth}
\begin{tikzpicture}[scale=0.6]

    \node[below left] at (-4,0) {$c_1$};
    \draw[fill=black] (-4,0) circle (0.15);

    \node[below right] at (4,0) {$c_2$};
    \draw[fill=black] (4,0) circle (0.15);

    \node[below right] at (4,8) {$r_2$};
    \draw[fill=blue] (4,8) circle (0.15);

    \node[below left] at (-4,8) {$r_1$};
    \draw[fill=blue] (-4,8) circle (0.15);
    
    \draw[postaction={on each segment={mid arrow}}] (-4,0) -- (-2,2) -- (-2,6) -- (-4,8);
    \draw[postaction={on each segment={mid arrow}}] (4,0) -- (2,2) -- (2,6) -- (4,8);
    \draw[postaction={on each segment={mid arrow}}] (-2,2) -- (0,4) -- (2,6);
    \draw[postaction={on each segment={mid arrow}}] (2,2) -- (0,4) -- (-2,6);
    
    \draw[dashed] (2,2) -- (0,0) -- (-2,2);
    \node[below] at (0,0) {$\ket{\Psi}$};
    
    \draw[fill=yellow] (-2,2) circle (0.3);
    \draw[fill=yellow] (2,2) circle (0.3);
    \draw[fill=yellow] (-2,6) circle (0.3);
    \draw[fill=yellow] (2,6) circle (0.3);
    
    \draw[fill=gray,opacity=0.5] (-1,1) -- (1,1) -- (1,7) -- (-1,7) -- (-1,1);
    
\end{tikzpicture}
\caption{}
\label{fig:taggingsub2}
\end{subfigure} 
\caption{A position-verification scheme in $1+1$ dimensions. Inputs are given at locations $c_1$, $c_2$. The prover should apply a designated quantum operation to these inputs, then return the outputs to points $r_1$, $r_2$. a) An honest prover enters the designated spacetime region (grey) to apply the needed quantum operation. b) A dishonest prover attempts to reproduce the same operation while acting outside the spacetime region.}
\label{fig:2dsetup}
\end{figure}
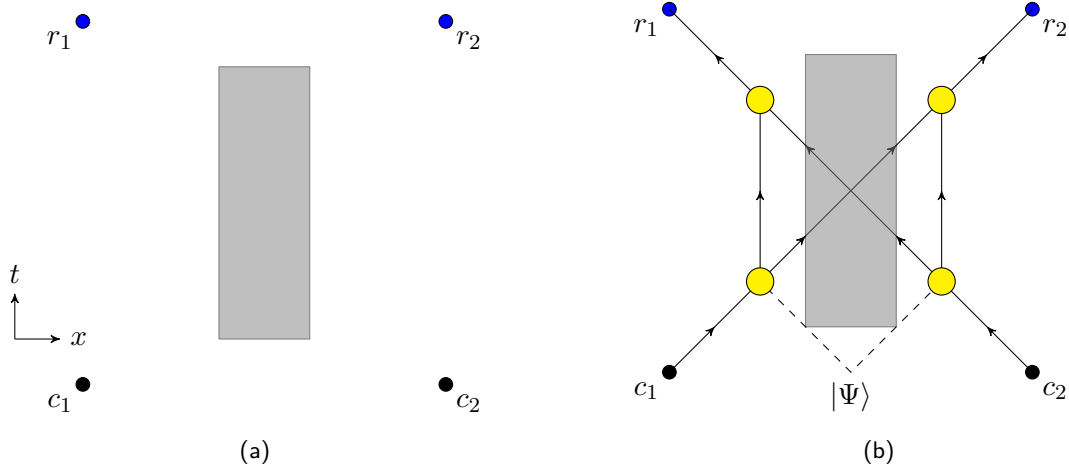

%%%%%%%%%%%%%%%%%%%%%%%%%%%%%%%%%%%%%%%%%%%%%%%%%%%%%%%%
\subsection{Prior work}
%%%%%%%%%%%%%%%%%%%%%%%%%%%%%%%%%%%%%%%%%%%%%%%%%%%%%%%%

We focus on the $f$-BB84 QPV scheme, which is the subject of the recent experiment \cite{fan2026relativistic} and the most well-studied scheme.  
An $f$-BB84 instance is defined by a choice of Boolean function $f:\{0,1\}^n\times \{0,1\}^n\rightarrow \{0,1\}$, along with a (fixed size) Hilbert space $\mathcal{H}_Q$, usually taken to be a single qubit. 
In the scheme, the prover is given inputs $x\in\{0,1\}^n, \ket{\psi}_Q$ at one spacetime location, and input $y\in\{0,1\}^n$ at a second spacetime location. 
See \cref{fig:2dsetup} for the spacetime set-up. 
To complete the verification task, the prover should bring the inputs together at a single location, compute $f(x,y)$, then measure system $Q$ in the computational basis if $f(x,y)=0$, or in the Hadamard basis if $f(x,y)=1$. 
The resulting measurement outcome should then be sent to two output locations. 
The placement of where the inputs are sent from, and where the outputs should be returned to, is chosen such that an honest prover operating as above must perform the measurement within a chosen spacetime region. 

The basic security intuition is that if a coalition of attackers is located outside the chosen spacetime region, the qubit and the classical information determining its measurement basis are initially split between these two attackers. The attacker who receives $Q$ also knows $x$ but cannot determine the correct measurement basis without the distant input $y$. Since the timing constraints prevent the attackers from bringing these inputs together, they must reproduce the required measurement non-locally using shared quantum resources.

The $f$-BB84 scheme is experimentally favourable for a number of reasons. 
Most prominently, the honest player need only execute $O(1)$ quantum gates, in particular a single-qubit measurement, classically controlled to be in either the computational or Hadamard bases. 
Meanwhile, it was long expected that the dishonest player needs quantum resources that grow with $n$, the classical input size. 
Indeed, all known cheating protocols have this property, in that they use both a number of quantum gates and a number of shared EPR pairs that grow with $n$, at least for suitable choices of function $f$. 

Significant efforts have gone towards showing that all attacks on $f$-BB84 schemes require growing quantum resources \cite{bluhm2021position, escola2023single, asadi_et_al:LIPIcs.ITCS.2025.11, asadi2025linear, allerstorfer2023making}. 
Recently, it has been understood that there is a possible obstruction towards a proof that the entanglement required grows quickly with $n$. 
In particular, the works \cite{bluhm2026complexity,bluhm2026equivalence} show that entanglement cost in $f$-BB84 is equivalent, up to small constant-factor overheads, to entanglement cost in a similar protocol known as $f$-routing. 
Meanwhile entanglement cost in $f$-routing itself is known to lower bound the randomness cost in the conditional disclosure of secrets (CDS) primitive studied in information-theoretic classical cryptography \cite{allerstorfer2023relating}. 
Finally, proving super-logarithmic lower bounds on CDS is an outstanding open problem in classical cryptography, when we consider explicit choices of function $f$ and allow the protocol to work only up to small errors \cite{applebaum2021placing}. 
Taken together, these implications mean that giving super-logarithmic lower bounds on entanglement cost in $f$-BB84, when allowing small errors and using an explicit function $f$, would imply a solution to a long-standing open problem in the classical literature. 
This suggests looking for another approach to addressing the security of $f$-BB84 which evades this apparent obstruction. 

One approach is to focus on random choices of function $f$; by not using an explicit function this evades the CDS barrier. 
This approach is used in a number of works \cite{bluhm2021position, escola2023single}. 
However, random functions have, with high probability, exponential complexity. 
In the recent experiment \cite{fan2026relativistic} a random function was used, which was implemented by storing the truth table of the function in a memory and rapidly accessing this memory to compute the function. 
However, this memory must be of size $2^{2n}$ for classical inputs of size $n$, so this quickly becomes infeasible as $n$ grows. 
As well, so far no linear bound on entanglement is known even for random functions. 
Instead, the bounds in \cite{bluhm2021position,escola2023single} are on the quantum ancilla size used by the attacker, and even then on the ancilla size in a purified model. 
This purified model may be counting classical memory used by a physical attacker as quantum ancilla, so it is not clear if this establishes a separation in the quantum resources needed by an honest or dishonest player. 

Another approach is initiated in \cite{asadi2025linear}, which lower bounds the quantum gates used by an attacker, rather than the entanglement cost. 
This means a linear lower bound can be achieved without facing a barrier from classical cryptography. 
This is done even when allowing small errors, and for simple, explicit, choices of function, for instance the inner product function. 
The inner product function is advantageous for QPV because it can be computed in logarithmic depth. 
The low-depth computation is especially important since the time taken to compute $f$, which is proportional to the depth, puts a lower bound on the size of the region that the prover can be localized to by the scheme. 

In the context of proving lower bounds for random functions, recent work has addressed whether the $f$-BB84 scheme remains secure in the presence of large noise. 
While for arbitrary types of noise the security proofs only allow for small errors, the schemes turn out to be much more tolerant of a particular form of noise known as loss \cite{escola2023single}. 
In the lossy setting, the honest player is allowed three possible outputs, $b\in\{0,1,\perp\}$, where the symbol $\perp$ amounts to declaring that the qubit $Q$ sent to them was lost, and so they are unable to respond. 
We call the probability of a loss $L$. 
The probability of responding incorrectly, conditioned on not declaring a loss, is labelled $\epsilon$. 
Potentially, for small $\epsilon$ the $f$-BB84 scheme may remain secure even at large values of $L$.
In the context of ancilla lower bounds against random functions, \cite{escola2023single} shows this.
Their approach involves studying a lossy version of a monogamy-of-entanglement game and using semidefinite programming techniques to numerically address the maximal winning probability of such games.

It is clear that the $f$-BB84 scheme is insecure above a loss probability $L=1/2$. 
This is because there is an attack that completes the task with loss $1/2$, using no shared entanglement and $O(1)$ quantum gates. 
In particular the attacker may guess the basis, measure, broadcast the measurement outcome to both sides, and then respond with the (correct) measurement outcome only if they turn out to have guessed the basis correctly. 
In this attack they declare a loss with probability of $L=1/2$ and error $\epsilon=0$. 
In practice, over large distances loss in optical fibres is much larger than $1/2$, so this limits the $f$-BB84 scheme to low distance implementations. 
To go beyond this, modifications of the $f$-BB84 scheme involving a commitment step \cite{allerstorfer2023making} have been proposed, which allows for loss tolerance up to arbitrarily large values of $L$.
Currently, implementing this commitment step has not been achieved experimentally, and to address current experiments it suffices to study the partial loss tolerance setting.

%%%%%%%%%%%%%%%%%%%%%%%%%%%%%%%%%%%%%%%%%%%%%%%%%%%%%%%%
\subsection{Summary of our results}
%%%%%%%%%%%%%%%%%%%%%%%%%%%%%%%%%%%%%%%%%%%%%%%%%%%%%%%%

In this work we prove gate lower bounds for the $f$-BB84 QPV scheme, where $f$ can be a simple function (e.g. inner-product), allowing loss, error, and imperfections in the states prepared by the verifier. 
Specifically, define:
\begin{itemize*}
    \item $L:=\Pr[\text{both players declare a loss}]$
    \item $1-\epsilon:=\Pr[\text{both players respond correctly}\,|\, \text{at least one player responds}]$
    \item $\eta:=\max_{\theta,z}\frac{1}{2}\Vert \rho^{\theta}_z-V^\theta_z\Vert_1$, where $V^\theta_z=H^\theta\ketbra{z}{z}H^\theta$, and $\rho^\theta_z$ are the states input by the referee to the task, when the intended basis is $\theta$ and intended correct answer is $z$. 
\end{itemize*}
The security region in the parameter space $(L,\epsilon)$ is given analytically by the inequality ($s=\sin^2(\pi/8)$)
\begin{align}
    \frac{\epsilon(1-L)}{s} + 2L+\frac{\eta}{s} < 1.
\end{align}
This region is illustrated in \cref{fig:securityregion}, where we fix $\eta$ and show the resulting secure $(\epsilon,L)$ region.
This is the region where, considering the inner-product function, the attacker needs to use $\tilde{\Omega}(n)$ quantum gates to attack the scheme, while an honest player needs $O(1)$ quantum gates.
In more detail we prove the lower bound
\begin{align}
    (\log(q)+1)(2C_G(f) + C_M(f)) \geq n - O(1)
\end{align}
where $q$ is the number of qubits used by Alice and Bob together, $C_G$ is the number of (at most 2 qubit) gates they apply drawn from a gate set of size 4, and $C_M$ is the number of single-qubit computational basis measurements.
Our bound can be adapted to other architectures, e.g. allowing larger gate sets or multi-qubit measurements.
Our final result is stated more formally as \cref{corrolary:IPapplication}. 

\begin{figure}
    \centering
    \includegraphics[width=0.5\linewidth]{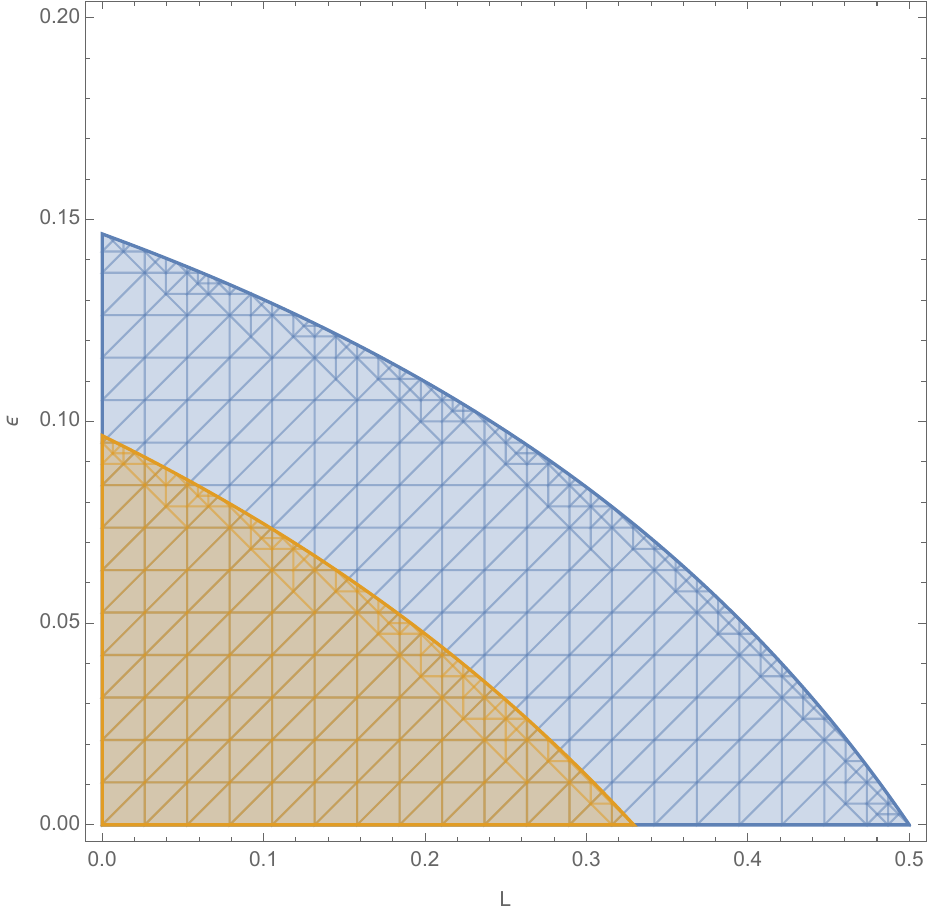}
    \caption{The parameter region in which we prove that the $f$-BB84 scheme requires nearly linear quantum gates to attack. We assume $f$ has linear SMP complexity; the inner product function for instance suffices. Here $L$ is the probability with which the player in the protocol is allowed to declare a loss. Meanwhile $1-\epsilon$ is the probability the response is correct, conditioned on at least one player not responding ``loss''. Our security region is as large as possible: outside the region there is an attack using $O(1)$ quantum gates. In blue we show the ideal case, where the quantum inputs to $f$-BB84 protocol are prepared in exactly the BB84 states. When the inputs are themselves noisy, the region shrinks. In yellow we show the security region when the inputs are $\eta=0.05$ close in trace distance to the ideal case.}
    \label{fig:securityregion}
\end{figure}

To prove our result, we study monogamy-of-entanglement (MoE) games which allow the players to declare loss, and where imperfections in the referee's inputs are allowed. 
In a lossy MoE game, the players, Alice and Bob prepare a tripartite state $\rho_{\bar{Q}AB}$, give $\bar{Q}$ to the referee, and then take $A$ and $B$ respectively. 
The players separate and are no longer able to communicate.
Then, the referee randomly chooses $\theta\in\{0,1\}$, measures $\bar{Q}$ in the computational basis if $\theta=0$ or Hadamard basis if $\theta=1$, and then reveals the choice of basis $\theta$ to the players.  
The players respond with a value from $\{0,1,\perp\}$ where $\perp$ signifies the player is declaring a loss.
Similar to the above definitions, we call the probability they are both correct $C$, and the probability they both declare a loss $L$. 
In \cref{thm:optimal_lossy_bb84_bound}, we show that
\begin{align}
    C+\frac{L}{\sqrt{2}} \leq \bc.
\end{align}
Compared to earlier work \cite{escola2023single}, we obtain this result analytically, rather than using a numerical SDP optimization, and without assuming the players always return the same response.

Using the lossy monogamy game bound, we give a reduction from any $f$-BB84 protocol to the simultaneous message passing (SMP) communication scenario. 
In the SMP scenario, two non-communicating players, Alice and Bob, get inputs $x\in\{0,1\}^n$ and $y\in\{0,1\}^n$, respectively. 
They each send a message to a referee, who should compute $f(x,y)$. 
We find that the number of quantum gates + single-qubit measurements performed by the players in the first round of the $f$-BB84 protocol upper bounds the communication cost in the SMP scenario. 
See \cref{fig:reductionfig}. 

\begin{figure*}
    \centering
    \begin{subfigure}{0.45\textwidth}
    \centering
    \begin{tikzpicture}[scale=0.45]
    
    %lower left box
    \draw[thick] (-5,-5) -- (-5,-3) -- (-3,-3) -- (-3,-5) -- (-5,-5);
    \node at (-4,-4) {$\mathcal{N}^x$};
    
    %lower right box
    \draw[thick] (5,-5) -- (5,-3) -- (3,-3) -- (3,-5) -- (5,-5);
    \node at (4,-4) {$\mathcal{M}^y$};
    
    %top right box
    \draw[thick] (5.5,5) -- (5.5,3) -- (2.5,3) -- (2.5,5) -- (5.5,5);
    
    %top left box
    \draw[thick] (-5.5,5) -- (-5.5,3) -- (-2.5,3) -- (-2.5,5) -- (-5.5,5);
    
    %left vertical wire
    \draw[thick, mid arrow] (-4.5,-3) -- (-4.5,3);
    \node[left] at (-4.5,-2) {$M_0$};
    
    %right vertical wire
    \draw[thick, mid arrow] (4.5,-3) -- (4.5,3);
    \node[right] at (4.5,-2) {$M_1'$};
    
    %left to right wire
    \draw[thick, mid arrow] (-3.5,-3) to [out=90,in=-90] (3.5,3);
    \node[right] at (-3.25,-2) {$M_0'$};
    
    %right to left wire
    \draw[thick, mid arrow] (3.5,-3) to [out=90,in=-90] (-3.5,3);
    \node[left] at (3.25,-2) {$M_1$};
    
    %entanglement
    \draw[thick] (-3.5,-5) to [out=-90,in=-90] (3.5,-5);
    \draw[black] plot [mark=*, mark size=3] coordinates{(0,-7.05)};
    
    %input wires
    \draw[thick] (-4.5,-6) -- (-4.5,-5);
    \node[below] at (-4.5,-6) {$Q,x$};
    \draw[thick] (4.5,-6) -- (4.5,-5);
    \node[below] at (4.5,-6) {$y$};
    
    %output wires
    \draw[thick] (4.5,5) -- (4.5,6);
    \node[above] at (-4.5,6) {$a$};
    \draw[thick] (-4.5,5) -- (-4.5,6);
    \node[above] at (4.5,6) {$b$};
    
    \end{tikzpicture}
    \caption{}\label{fig:oneway}
    \end{subfigure}
    \hfill
    \begin{subfigure}{0.45\textwidth}
    \centering
    \begin{tikzpicture}[scale=0.4]
    
    %lower left box
    \draw[thick] (-5,-5) -- (-5,-3) -- (-3,-3) -- (-3,-5) -- (-5,-5);
    
    %lower right box
    \draw[thick] (5,-5) -- (5,-3) -- (3,-3) -- (3,-5) -- (5,-5);
    
    %top right box
    \draw[thick] (5,5) -- (5,3) -- (3,3) -- (3,5) -- (5,5);
    
    %right vertical wire
    \draw[thick, mid arrow] (4,-3) -- (4.5,3);
    
    %left to right wire
    \draw[thick, mid arrow] (-4,-3) to [out=90,in=-90] (3.5,3);
    
    %entanglement
    \draw[thick] (-3.5,-5) to [out=-90,in=-90] (3.5,-5);
    \draw[black] plot [mark=*, mark size=3] coordinates{(0,-7.05)};
    
    %input wires
    \draw[thick] (-4.5,-6) -- (-4.5,-5);
    \node[below] at (-4.5,-6) {$x$};
    
    \draw[thick] (4.5,-6) -- (4.5,-5);
    \node[below] at (4.5,-6) {$y$};
    
    %output wires
    \draw[thick] (4,5) -- (4,6);
    \node[above] at (4,6) {$f(x,y)$};
    
    \end{tikzpicture}
    \caption{}
    \label{fig:simultaneous}
    \end{subfigure}
    \caption{a) An $f$-BB84 protocol. Alice on the left applies the first round operation $\mathcal{N}^x$; Bob on the right applies the first round operation $\mathcal{M}^y$. We define $M=M_0M_1$, $M'=M_0'M_1'$. b) An SMP protocol, which can be defined from the $f$-BB84 protocol. The messages sent from Alice and Bob to the referee contain a description of their first round operations. The message lengths are related to the number of quantum gates the players apply. To ensure only quantum gates are counted, it is necessary to give Alice and Bob the same entangled state as they share in the $f$-BB84 protocol.} 
    \label{fig:reductionfig}
\end{figure*}
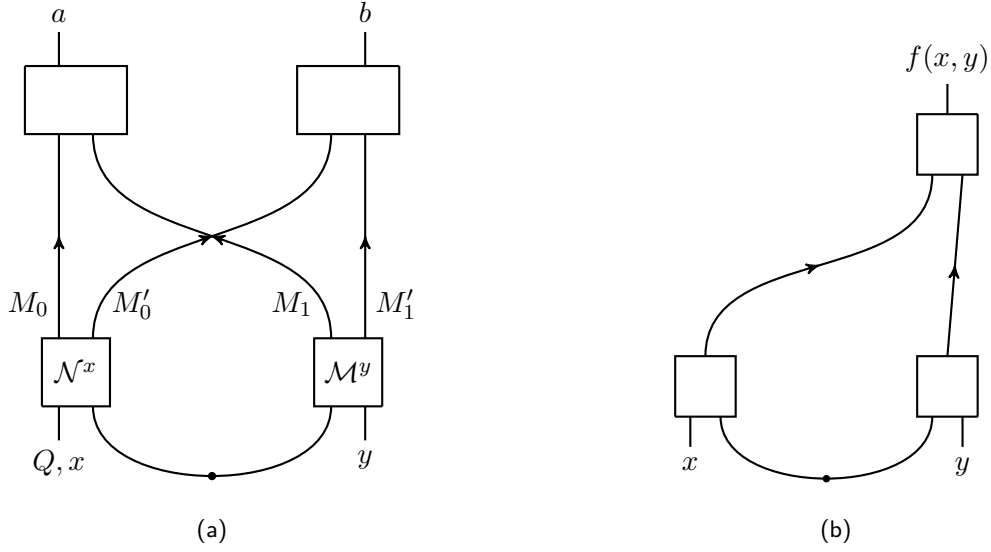

The basic idea underlying the reduction is as follows. 
In the SMP protocol, Alice and Bob send a description of the operations they would use in the first round of the $f$-BB84 protocol. 
The referee then computes the density matrix $\rho_{\bar{Q}MM'}^{x,y}$ from their descriptions. 
The density matrix $\rho_{\bar{Q}MM'}^{x,y}$ needs to allow Alice, who holds $M$, and Bob, who holds $M'$, to obtain the same measurement outcome as the referee, who measures $\bar{Q}$ in either the computational or Hadamard basis, with the choice of basis fixed by $f(x,y)$. 
Our bound on MoE games then tells us that $\rho_{\bar{Q}MM'}^{x,y}$ can only be good, even allowing the players to declare loss, at completing this task for one choice of measurement. 
The referee determines which basis the density matrix $\rho$ is better for, and then biases their response towards the corresponding value of $f(x,y)$. 

While this is the basic approach, a number of complications come up in formalizing this idea. 
First, in practice the players in the $f$-BB84 protocol may use mid-circuit measurement outcomes and classical processing to adaptively choose later quantum gates in their circuit, as is standard in most proposed quantum computing architectures. 
This means Alice and Bob's operations are not fixed given $(x,y)$. 
To handle this, we have Alice and Bob in the SMP protocol share the same entangled state as is used in the $f$-BB84 protocol, and run their procedure as they would in the $f$-BB84 protocol. 
They then obtain a sample from the set of possible local operations they apply in the $f$-BB84 protocol, and they send this sample to the referee. 
This sample fixes a density matrix $\rho_{\bar{Q}MM'}^{x,y}(m)$ where $m$ represents the measurement outcomes obtained in this run, and $\rho_{\bar{Q}MM'}^{x,y}=\sum_m p_m \rho_{\bar{Q}MM'}^{x,y}(m)$. 
Thus in practice the referee looks at $\rho_{\bar{Q}MM'}^{x,y}(m)$ and biases their answer towards the value of $f(x,y)$ for which this density matrix works better --- we show this biases the referee towards responding correctly. 

A second complication comes from imperfections in how the BB84 states used in the protocol are prepared. 
An inspection of the proof outlined above reveals that it works only when the input states are exactly the BB84 states, which allows us to view the game in a purified view where the input is the $Q$ subsystem of the maximally entangled state $\Psi^+_{\bar{Q}Q}$.
However, we find we can straightforwardly handle imperfections in the input states by supplementing the proof with perfect inputs with a continuity statement on the success and loss probabilities of the $f$-BB84 game.

Aside from extending the gate lower bounds of \cite{asadi2025linear} to the lossy case, even at zero loss our techniques improve on those of \cite{asadi2025linear}, in particular extending the allowed error probability from $\epsilon=0.055$ to $\epsilon=0.146$. 
In fact, our security region is as large as possible, at least when considering ideal BB84 state inputs: everywhere outside the region there is a $O(1)$ quantum gate attack. 

%%%%%%%%%%%%%%%%%%%%%%%%%%%%%%%%%%%%%%%%%%%%%%%%%%%%%%%%
\section{Background and tools}
%%%%%%%%%%%%%%%%%%%%%%%%%%%%%%%%%%%%%%%%%%%%%%%%%%%%%%%%

%%%%%%%%%%%%%%%%%%%%%%%%%%%%%%%%%%%%%%%%%%%%%%%%%%%%%%%%
\subsection{Distance measures and entropy inequalities}
%%%%%%%%%%%%%%%%%%%%%%%%%%%%%%%%%%%%%%%%%%%%%%%%%%%%%%%%

In this section, we give a few definitions and collect some standard results for reference. 
We label the dimension of a Hilbert space $\mathcal{H}_A$ as $d_A$ and define $\log d_A=n_A$. 
Unless otherwise specified, all logarithms are taken base 2. 

We make use of the one-norm, 
\begin{align}
    \Vert A \Vert_1 = \tr\sqrt{A^\dagger A}.
\end{align}
The trace distance is defined by
\begin{align}
    D(\rho,\sigma)=\frac{1}{2}\Vert \rho-\sigma \Vert_1.
\end{align}
We will make use of the Holevo-Helstrom theorem, which states the following. 
Suppose that we are given a state described by $\rho$, or a state described by $\sigma$, each with probability $1/2$. 
Then the maximal probability of correctly determining which state was given is
\begin{align}
    p_{\text{dist}}^{\text{opt}}(\rho,\sigma)=\frac{1}{2}+\frac{1}{2}\left(\frac{1}{2}\Vert\rho-\sigma \Vert_1 \right).
\end{align}

%%%%%%%%%%%%%%%%%%%%%%%%%%%%%%%%%%%%%%%%%%%%%%%%%%%%%%%%
\subsection{Communication complexity}
%%%%%%%%%%%%%%%%%%%%%%%%%%%%%%%%%%%%%%%%%%%%%%%%%%%%%%%%

We will make use of a reduction from $f$-routing and $f$-BB84 to communication complexity scenarios. 
Specifically, we will be interested in the \emph{simultaneous message passing} ($\SMP$) scenario. 

A simultaneous message passing scenario is defined by a choice of function $f:\{0,1\}^n\times \{0,1\}^n\rightarrow \{0,1\}$. 
The scenario involves three parties, Alice, Bob, and the referee. 
Alice receives $x\in \{0,1\}^n$ and Bob receives $y\in \{0,1\}^n$. 
Alice and Bob compute messages $m_A, m_B$ from their local resources (including shared randomness) and the inputs they receive, and send their messages to the referee. 
Alice and Bob succeed if the referee can compute $f(x,y)$ from their messages. 
We define the $\SMP$ cost of $f$, denoted $\SMP(f)$ to be $\min_P( |m_A|+|m_B|)$ where the minimization is over choices of protocols.

A formal definition follows.

\begin{definition}[SMP complexity]\label{def:SMP}
Let $f : \{0,1\}^n\times \{0,1\}^n\rightarrow \{0,1\}$ be a function, and $\epsilon\in [0,1]$ be a parameter. An $\SMP$ protocol $P$ for $f$ consists of three algorithms Alice, Bob, and a referee. Alice receives $x\in \{0,1\}^n$ as input and outputs $m_A \in \{0,1\}^*$, Bob receives $y\in \{0,1\}^n$ as input and outputs $m_B \in \{0,1\}^*$, and the referee receives $m_A, m_B$ and outputs a bit $c=P(x,y)$. 
A protocol $P$ is $\epsilon$-correct on the input distribution $\mu$ if
\begin{align*}
    \Pr_{\mu,R}[P(x,y)=f(x,y)] \geq 1-\epsilon \enspace.
\end{align*}
The $\epsilon$-SMP complexity of $f$ is defined as follows
\begin{equation*}
    \SMP_{\mu,\epsilon}(f) = \min_{P: P \text{ is $\epsilon$-correct}}(\abs{m_A} + \abs{m_B}) \enspace.
\end{equation*}
Similarly, we can define $\SMP^*_{\mu,\epsilon}(f)$ for the case where Alice and Bob share entanglement.
\end{definition}

A standard function studied in communication complexity is the inner product, 
\begin{align*}
    \IPfunc_n(x,y) = \sum_{i=1}^n x_i y_i \,\, \text{mod}\,\, 2 \enspace.
\end{align*}
Intuitively, this is a difficult function to compute in communication complexity scenarios because the output depends sensitively on every bit of the input. 
More concretely, we will make use of the following lower bound, proven in \cite{nayak2002communication}:
\begin{align}\label{eq:IPlowerbound}
    Q_{u,\epsilon}^*(\IPfunc_n) \geq \frac{1}{2}n + \log(1-2\epsilon)
\end{align}
where $Q_{u,\epsilon}^*(f)$ denotes the two way quantum communication complexity of $f$, meaning the number of qubits needed to compute $f$ with probability at least $1-\epsilon$, over the uniform distribution $u$, when Alice and Bob are allowed shared entanglement. 
Another lower bound on the inner product function is \cite{cleve1998quantum} 
\begin{align}\label{eq:clevebound}
    R^*_{u,\epsilon}(\IPfunc) \geq \max \left\{\frac{1}{2}(1-2\epsilon)^2, (1-2\epsilon)^4\right\}n-1/2.
\end{align}
Here $R^*_{u,\epsilon}$ denotes the two-way classical communication complexity, with shared entanglement allowed. 
In our explicit bounds, we will prefer \cref{eq:IPlowerbound} over \cref{eq:clevebound} for simplicity, but notice that \cref{eq:clevebound} is stronger in the low error regime, and could be used instead. 

To apply these bounds to the SMP model, notice that
\begin{align}\label{eq:RandQ}
    R_{\mu,\epsilon}^*(f)\geq Q_{\mu,\epsilon}^*(f)
\end{align}
since classical communication cannot be stronger than quantum communication. 
As well, note that the one-way classical communication complexity, with shared entanglement allowed, is lower bounded by the two-way communication complexity, 
\begin{align}\label{eq:1wayandR}
    R1^{*}_{\mu,\epsilon, A\rightarrow B}(f)\geq R_{\mu,\epsilon}^*(f),\nonumber \\
    R1^{*}_{\mu,\epsilon, B\rightarrow A}(f)\geq R_{\mu,\epsilon}^*(f)
\end{align}
Here the $A\rightarrow B$ subscript indicates we suppose Alice sends a message to Bob but not vis versa, and $B\rightarrow A$ similarly means Bob only sends a message to Alice. 
We can also notice that 
\begin{align}\label{eq:SMPandR}
    \SMP^*_{\mu,\epsilon}(f) \geq R1_{\mu,\epsilon,A\rightarrow B}^*(f)+R1_{\mu,\epsilon,B\rightarrow A}^*(f).
\end{align}
This follows because the SMP cost is defined as the sum of Alice and Bobs message lengths; we can notice that if instead of sending her message to the referee she sends it to Bob, he can play the role of the referee and compute $f$. 
Thus the message length $m_A$ in the SMP protocol is an upper bound on $ R1_{\mu,\epsilon,A\rightarrow B}^*(f)$. 
Similarly, $m_B$ is an upper bound on $ R1_{\mu,\epsilon,B\rightarrow A}^*(f)$. 
Combining \cref{eq:SMPandR}, \cref{eq:1wayandR}, \cref{eq:RandQ}, and \cref{eq:IPlowerbound}, we find that
\begin{align}
    \boxed{\SMP^*_{u,\epsilon}(\IPfunc_n)\geq n + 2\log(1-2\epsilon).}
\end{align}

%%%%%%%%%%%%%%%%%%%%%%%%%%%%%%%%%%%%%%%%%%%%%%%%%%%%%%%%
\section{Monogamy-of-entanglement games}
%%%%%%%%%%%%%%%%%%%%%%%%%%%%%%%%%%%%%%%%%%%%%%%%%%%%%%%%

One of our main technical tools for analyzing the $f$-BB84 QPV scheme will be monogamy-of-entanglement (MoE) games.  In particular, we improve upon the analysis of lossy MoE games in \cite{escola2023single} in that we reproduce the output of their numerical SDP analysis with an analytical proof. We also improve on their analysis by relaxing some restrictions they place on the players. In particular, we do not require that the players always return the same output.

We first define the notion of a lossy monogamy-of-entanglement game. 

\begin{definition}\label{def:lossyMoE}
    A lossy BB84 monogamy-of-entanglement game consists of two stages, a state preparation phase and then measurement phase. 
    \begin{itemize}
        \item \textbf{Preparation phase:} Alice and Bob prepare an arbitrary state $\Psi_{\bar{Q}AB}$ with $\bar{Q}$ a single qubit, and give $\bar{Q}$ to the referee. Alice takes system $A$, Bob system $B$, and then the players separate and can no longer communicate. 
        \item \textbf{Measurement phase:} The referee chooses $\theta\in\{0,1\}$ uniformly at random, then measures $\bar{Q}$ in the computational basis if $\theta=0$, or the Hadamard basis if $\theta=1$. She labels her measurement outcome $z$. The referee then reveals $\theta$ to both Alice and Bob. Alice and Bob then respond with values $a,b$ from the set $\{0,1,\perp\}$, where $\perp$ denotes `loss'.
    \end{itemize} 
    We define $L:=\Pr[a=b=\perp]$, $C:=\Pr[a=b=z]$.
\end{definition}

Before proving an upper bound on the success probability in the lossy setting, we first give an alternative proof of the well-known (tight) upper bound on the non-lossy ($L=0$) BB84 MoE.
Our strategy is similar to \cite{johnston2016extended}.

\begin{proposition}
    Considering the non-lossy ($L=0$) BB84 MoE game, we have $C\leq \bc$.
\end{proposition}

\begin{proof}\,
    We define the measurement operators
    \begin{align}
        \Pi^0 &= \sum_x V_x^0 \otimes A_x^0 \otimes B_x^0, \\
        \Pi^1 &= \sum_x V_x^1 \otimes A_x^1 \otimes B_x^1,
    \end{align}
    where $V_x^\theta \coloneq H^\theta \ketbra{x}{x} H^\theta$.
    Similar to \cite{tomamichel2013monogamy} we introduce a `weakened' operator:
    \begin{align}
        \tilde{P} &= \sum_x V_x^0 \otimes A_x^0 \otimes \mathbbm{1}_B \\
        \tilde{Q} &= \sum_y V_y^1 \otimes \mathbbm{1}_A \otimes B_y^1.
    \end{align}
    Note:
    \begin{align}
        &\Pi^0 \preceq \tilde{P}, &&\Pi^1 \preceq \tilde{Q}, &\frac{\Pi^0 + \Pi^1}{2} \preceq \frac{\tilde P + \tilde Q}{2},
    \end{align}
    since $(\mathbbm{1}_A - A^0_x) \succeq 0$, etc. 

    We will use $\sum_x A_x^0 = \mathbbm{1}_A, \sum_y B_y^1 = \mathbbm{1}_B$, then:
    \begin{align}
        \tilde P &= \sum_x V_x^0 \otimes A_x^0 \otimes \left( \sum_y B_y^1 \right) = \sum_{x,y} V_x^0 \otimes A_x^0 \otimes B_y^1 \\
        \tilde Q &= \sum_y V_y^1 \otimes \left( \sum_x A_x^0 \right) \otimes B_y^1 = \sum_{x,y} V_y^1 \otimes A_x^0 \otimes B_y^1 \\
        \frac{\tilde P + \tilde Q}{2} &= \sum_{x,y} \frac{V_x^0 + V_y^1}{2} \otimes A^0_x \otimes B_y^1.
    \end{align}
    For every pair of $x,y$ the largest eigenvalue of $\frac{V_x^0 + V_y^1}{2}$ is $\bc$, therefore
    \begin{align}
        \frac{V_x^0 + V_y^1}{2} &\preceq \bc \mathbbm{1}_Q.
    \end{align}
    Then, summing over $x,y$ gives:
    \begin{align} 
        \frac{\tilde P + \tilde Q}{2} &\preceq \bc \mathbbm{1}_Q \otimes \left( \sum_x A_x^0 \right) \otimes \left( \sum_y B_y^1 \right) = \bc \mathbbm{1}_{QAB}.
    \end{align}
    Thus,
    \begin{align}
        C &= \tr[ \frac{\Pi^0 + \Pi^1}{2} \rho_{QAB} ] \leq \tr[ \frac{\tilde P + \tilde Q}{2} \rho_{QAB}] \\ &\leq \tr[ \bc \mathbbm{1}_{QAB} \, \rho_{QAB}] = \bc.
    \end{align}
\end{proof}

Next we use a similar line of reasoning for the MoE game with loss, where the difference is that we will also have to introduce additional operators for the attackers claiming loss. 
In the case of loss Alice and Bob perform three-outcome POVMs with possible outcomes $\mathcal{O} = \{0,1,\perp\}$. 
We define the joint-correctness and joint-abort probabilities operators and their respective expectations:
\begin{align}
    \tilde C_\theta = \sum_{x=0}^{1} (V_x^\theta \otimes A_x^\theta \otimes B_x^\theta), \qquad C_\theta &= \tr[ \rho_{QAB} \tilde{C}_\theta ], \\
    \tilde L_\theta = (\mathbbm{1}_Q \otimes A^\theta_\perp \otimes B^\theta_\perp), \qquad L_\theta &= \tr[ \rho_{QAB} \tilde{L}_\theta].
\end{align}
Averaged over the basis we define:
\begin{align}
    &\tilde C = \frac{\tilde C_0 + \tilde C_1}{2}, &&C = \frac{C_0 + C_1}{2}, &&\tilde L = \frac{\tilde L_0 + \tilde L_1}{2}, &L = \frac{L_0 + L_1}{2}.
\end{align}
Thus $C$ and $L$ are unconditional probabilities. In particular,
the definition does not assume that the players' responses always agree.

\begin{theorem}[Optimal lossy BB84 bound] \label{thm:optimal_lossy_bb84_bound}
Every state \(\rho_{QAB}\) and every collection of local three-outcome POVMs as
above satisfy
\begin{align}
    C+\frac{L}{\sqrt 2}\leq \cos^2\left(\frac{\pi}{8}\right).
\end{align}

\end{theorem}

\begin{proof} \,
We relax the operators $\tilde C_\theta, \tilde L_\theta$ again, like in the above proof for the BB84 MoE bound, and introduce `weakened loss operators':
    \begin{align}
        \tilde C_0 &\preceq \sum_{x=0}^{1} V_x^0 \otimes A_x^0 \otimes \mathbbm{1}_B, \\
        \tilde C_1 &\preceq \sum_{x=0}^{1} V_x^1 \otimes \mathbbm{1}_A \otimes B_x^1, \\
        \tilde L_0 &\preceq \mathbbm{1}_Q \otimes A_\perp^0 \otimes \mathbbm{1}_B, \\
        \tilde L_1 &\preceq \mathbbm{1}_Q \otimes \mathbbm{1}_A \otimes B_\perp^1.
    \end{align}
    Next, we define:
    \begin{align}
        \tilde T \coloneq \frac{1}{2} \left( \sum_{x=0}^{1} V_x^0 \otimes A_x^0 \otimes \mathbbm{1}_B + \sum_{x=0}^{1} V_x^1 \otimes \mathbbm{1}_A \otimes B_x^1 + \frac{\mathbbm{1}_Q}{\sqrt{2}} \otimes A_\perp^0 \otimes \mathbbm{1}_B +  \frac{\mathbbm{1}_Q}{\sqrt{2}} \otimes \mathbbm{1}_A \otimes B_\perp^1 \right),
    \end{align}
    such that
    \begin{align}
        \tilde C + \frac{\tilde L}{\sqrt{2}} &= \frac{1}{2} \left( \tilde C_0 + \tilde C_1 + \frac{\tilde L_0 + \tilde L_1}{\sqrt{2}} \right) \preceq \tilde T.
    \end{align}
    Now, we expand all the operators in $\tilde T$ using $\sum_{a \in \{0,1,\perp\}} A^0_a = \mathbbm{1}_A, \sum_{b \in \{0,1,\perp\}} B^1_b = \mathbbm{1}_B$, for instance the first term in $\tilde{T}$ is expressed as
    \begin{align}
        \sum_{a=0}^{1} V_a^0 \otimes A_a^0 \otimes \mathbbm{1}_B = \sum_{a=0}^{1} \sum_{b \in \{0,1,\perp\}} V_a^0 \otimes A_a^0 \otimes B_b^1 = \sum_{a,b \in \{0,1,\perp\}} \delta_{a\neq\perp} V_a^0 \otimes A_a^0 \otimes B_b^1,
    \end{align}
    where we inserted $\delta_{a\neq\perp}:=\delta_{a,0}+\delta_{a,1}$ to get $a,b$ both summed over $\{0,1,\perp\}$. 
    Then, continuing in this way for the remaining terms, we get 
    \begin{align}
        \tilde T = \sum_{a,b \in \{0,1,\perp\}} K_{ab} \otimes A_a^0 \otimes B_b^1, 
    \end{align}
    where we have defined the operator on $Q$:
    \begin{align}
        K_{ab} = \frac{1}{2} \left( \delta_{\{a\neq \perp\}} V_a^0 + \delta_{\{b \neq \perp\}} V_b^1 + \frac{1}{\sqrt{2}}(\delta_{\{a,\perp\}} + \delta_{\{b,\perp\}}) \mathbbm{1}_Q \right).
    \end{align}
    We now claim that 
    \begin{align}
        K_{ab} \preceq \bc \mathbbm{1}_Q, \qquad \text{for every } a,b \in \{0,1,\perp\}. 
    \end{align}
    There are 3 cases: both $a,b \in \{0,1\}$, exactly one output is $\perp$, or $a=b=\perp$.

    Suppose both $a,b \in \{0,1\}$ then:
    \begin{align}
        K_{ab} = \frac{V_a^0 + V_b^1}{2},
    \end{align}
    which has maximal eigenvalue $\bc$, thus $K_{ab} \preceq \bc \mathbbm{1}_Q$. 

    If exactly one output is $\perp$ then either:
    \begin{align}
        K_{\perp b} = \frac{1}{2} \left( \frac{\mathbbm{1}_Q}{\sqrt{2}} + V_b^1 \right) \quad \text{or} \quad K_{a \perp} = \frac{1}{2} \left( \frac{\mathbbm{1}_Q}{\sqrt{2}} + V_a^0 \right),
    \end{align}
    which both have the same largest eigenvalue $\frac{1}{2}\left( \frac{1}{\sqrt{2}} + 1\right) = \bc$, thus both $K_{\perp b}\preceq \bc \mathbbm{1}_Q$ and $K_{a \perp} \preceq \bc \mathbbm{1}_Q$.

    If both outputs are $\perp$ then:
    \begin{align}
        K_{\perp \perp} \coloneq \frac{\mathbbm{1}_Q}{\sqrt{2}} \preceq \bc \mathbbm{1}_Q,
    \end{align}
    since $1/\sqrt{2} \leq \bc$. This proves the claim.

    Now taking the inequality on $\tilde T$:
    \begin{align}
        \tilde T &= \sum_{a,b \in \{0,1,\perp\}} K_{ab} \otimes A_a^0 \otimes B_b^1 \preceq \bc \sum_{a,b \in \{0,1,\perp\}} \mathbbm{1}_Q \otimes A_a^0 \otimes B_b^1 \\
        &= \bc \mathbbm{1}_Q \otimes \left( \sum_{a \in \{0,1,\perp\}} A_a^0 \right) \otimes \left( \sum_{b \in \{0,1,\perp\}} B_b^1 \right) = \bc \mathbbm{1}_{QAB}.
    \end{align}
    Thus we have
    \begin{align}
        C + \frac{ L}{\sqrt{2}} = \tr[\left( \tilde C + \frac{\tilde L}{\sqrt{2}} \right) \rho_{QAB} ] \leq \tr[ \tilde T \rho_{QAB} ] \leq \tr[\bc \mathbbm{1}_{QAB} \ \rho_{QAB}] = \bc,
    \end{align}
    as claimed.
\end{proof}

This upper bound is also achievable at every value of $0\leq L\leq 1/2$, so that our bound is optimal. 
The following strategy achieves this: Alice and Bob share an EPR pair with the referee, $\Psi^+_{\bar{Q}Q}$. 
With probability $p$, Alice and Bob measure $Q$ (before separating) in the Breidbart basis, 
\begin{align}
    \ket{B_0} &= \cos(\pi/8)\ket{0}+\sin(\pi/8)\ket{1} \nonumber \\
    \ket{B_1} &= \sin(\pi/8)\ket{0}-\cos(\pi/8)\ket{1}
\end{align}
and both return the measurement outcome as their responses. 
This succeeds with probability $\cos^2(\pi/8)$ and never declares a loss. 
With probability $1-p$, Alice and Bob guess either the computational or Hadamard basis, measure in that basis, and then return their measurement outcome only if they guessed the correct basis. 
Otherwise, they declare a loss. 
This combined strategy achieves
\begin{align}
    C&=p\bc+\frac{1-p}{2} \nonumber \\
    L&=\frac{1-p}{2}
\end{align}
We can check that $C$, $L$ satisfy $C+L/\sqrt{2}=\bc$. 
As well, we can choose any value of $L$ in $[0,1/2]$ by varying $p$. 

We also remark that this technique can be extended to other games. 
\begin{remark}[Arbitrary overlap]
The only property of the computational and Hadamard bases used above is their
overlap. For two rank-one projective measurements
\(V_a^0=\lvert\psi_a^0\rangle\!\langle\psi_a^0\rvert\) and
\(V_b^1=\lvert\psi_b^1\rangle\!\langle\psi_b^1\rvert\), define
\begin{align}
    \gamma\coloneq
    \max_{a,b}\left\lvert\langle\psi_a^0\vert\psi_b^1\rangle\right\rvert.
\end{align}
The corresponding inequality is
\begin{align}
    C+\gamma L\leq\frac{1+\gamma}{2}.
\end{align}
\end{remark}

%%%%%%%%%%%%%%%%%%%%%%%%%%%%%%%%%%%%%%%%%%%%%%%%%%%%%%%%
\section{Analysis of lossy \texorpdfstring{$f$}{TEXT}-BB84}
%%%%%%%%%%%%%%%%%%%%%%%%%%%%%%%%%%%%%%%%%%%%%%%%%%%%%%%%

We are now ready to prove our lower bounds on the $f$-BB84 task, factoring in both loss and imperfect input preparation. 
The strategy is similar to \cite{asadi2025linear}: we prove a reduction from $f$-BB84 to the $\SMP$ scenario, where the communication cost in the $\SMP$ protocol is related to the number of quantum gates used in the $f$-BB84 protocol.
We begin in the next section by defining the $f$-BB84 task formally. 

%%%%%%%%%%%%%%%%%%%%%%%%%%%%%%%%%%%%%%%%%%%%%%%%%%%%%%%%
\subsection{Definition}\label{sec:defbb84}
%%%%%%%%%%%%%%%%%%%%%%%%%%%%%%%%%%%%%%%%%%%%%%%%%%%%%%%%

We give the following definition of a \emph{lossy $f$-BB84} task.
In this definition and the rest of the text, we refer to the two agents of the prover as Alice and Bob. 

\begin{definition}\label{def:qubitfbb84}
    An $\eta$-\textbf{ideal, lossy $f$-BB84} task is defined by a choice of Boolean function $f:\{0,1\}^{n}\times \{ 0,1\}^{n}\rightarrow \{0,1\}$, and a $2$ dimensional Hilbert space $\mathcal{H}_Q$.
    Inputs $x\in \{0,1\}^{n}$ and system $Q$ are given to Alice, and input $y\in \{0,1\}^{n}$ is given to Bob.
    The system $Q$ is in the state $\rho^{f(x,y)}_z$, with $z$ chosen uniformly at random, and $\frac{1}{2}\Vert \rho^{\theta}_z-V^{\theta}_z\Vert_1\leq \eta$, for all $\theta,z$.  
    Alice and Bob exchange one round of communication, with the combined systems received or kept by Alice labelled $M$ and the systems received or kept by Bob labelled $M'$.
    Alice and Bob can output values $a,b$ respectively with $a,b\in\{0,1,\perp\}$.
    We say the lossy $f$-BB84 task is completed $C$-correctly with loss $L$ on input $(x,y)$ if $C=\Pr[a=b=z|x,y]$ and $L=\Pr[a=b=\perp|x,y]$
\end{definition}

Next, we give a fully general model capturing strategies that complete the lossy $f$-BB84 task in the form of a non-local quantum computation.   

\begin{enumerate}
    \item Alice and Bob share a resource system $\Psi_{AB}$. 
    \item The referee prepares $(\rho_z^{f(x,y)})_Q$ and hands $Q$ to Alice. Alice also receives $x\in\{0,1\}^n$, while Bob receives $y\in\{0,1\}^n$.
    \item Alice applies $\mathcal{N}^x_{QA\rightarrow M_0M_0'}$, Bob applies $\mathcal{M}^y_{B\rightarrow M_1M_1'}$. Label $M=M_0M_1$, $M'=M_0'M_1'$ so that their joint state after the first round is 
    \begin{align*}
        \rho_{MM'}^{x,y}=\mathcal{N}^x_{QA\rightarrow M_0M_0'}\otimes \mathcal{M}^y_{B\rightarrow M_1M_1'}(\Psi_{QAB})\,.
    \end{align*} 
    \item $M_0$ and $M_1$ are sent to Alice, so that she holds $M$. At the same time, $M_0'$ and $M_1'$ are sent to Bob so that he holds $M'$. The (classical) inputs $x$ and $y$ are copied and sent to both parties. 
    \item Alice and Bob apply POVMs $\{\Lambda^{x,y,0}_{M},\Lambda^{x,y,1}_{M}, \Lambda^{x,y,\perp}_{M}\}$ and $\{\Lambda^{x,y,0}_{M'},\Lambda^{x,y,1}_{M'},\Lambda^{x,y,\perp}_{M'}\}$, then both output their measurement outcomes, which we label $a$, $b$, respectively. 
\end{enumerate}
Recall that Alice and Bob succeed when they both obtain outcome $z$. 
See \cref{fig:fBB84} for an illustration of a general protocol for $f$-BB84.

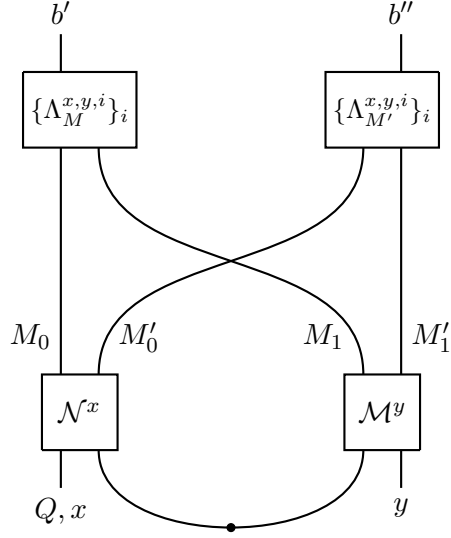
\begin{figure*}
    \centering
    \begin{tikzpicture}[scale=0.5]
    
    %lower left box
    \draw[thick] (-5,-5) -- (-5,-3) -- (-3,-3) -- (-3,-5) -- (-5,-5);
    \node at (-4,-4) {$\mathcal{N}^x$};
    
    %lower right box
    \draw[thick] (5,-5) -- (5,-3) -- (3,-3) -- (3,-5) -- (5,-5);
    \node at (4,-4) {$\mathcal{M}^y$};
    
    %top right box
    \draw[thick] (5.5,5) -- (5.5,3) -- (2.5,3) -- (2.5,5) -- (5.5,5);
    \node at (4,4) {\small{$\{\Lambda^{x,y,i}_{M'}\}_i$}};
    
    %top left box
    \draw[thick] (-5.5,5) -- (-5.5,3) -- (-2.5,3) -- (-2.5,5) -- (-5.5,5);
    \node at (-4,4) {\small{$\{\Lambda^{x,y,i}_M\}_i$}};
    
    %left vertical wire
    \draw[thick] (-4.5,-3) -- (-4.5,3);
    \node[left] at (-4.5,-2) {$M_0$};
    
    %right vertical wire
    \draw[thick] (4.5,-3) -- (4.5,3);
    \node[right] at (4.5,-2) {$M_1'$};
    
    %left to right wire
    \draw[thick] (-3.5,-3) to [out=90,in=-90] (3.5,3);
    \node[right] at (-3.25,-2) {$M_0'$};
    
    %right to left wire
    \draw[thick] (3.5,-3) to [out=90,in=-90] (-3.5,3);
    \node[left] at (3.25,-2) {$M_1$};
    
    %entanglement
    \draw[thick] (-3.5,-5) to [out=-90,in=-90] (3.5,-5);
    \draw[black] plot [mark=*, mark size=3] coordinates{(0,-7.05)};
    
    %input wires
    \draw[thick] (-4.5,-6) -- (-4.5,-5);
    \node[below] at (-4.5,-6) {$Q,x$};
    \draw[thick] (4.5,-6) -- (4.5,-5);
    \node[below] at (4.5,-6) {$y$};
    
    %output wires
    \draw[thick] (4.5,5) -- (4.5,6);
    \node[above] at (-4.5,6) {$b'$};
    \draw[thick] (-4.5,5) -- (-4.5,6);
    \node[above] at (4.5,6) {$b''$};
    
    \end{tikzpicture}
    \caption{A general strategy for an $f$-BB84 scheme. Alice applies $\mathcal{N}^x$ in the first round, Bob applies $\mathcal{M}^y$. They communicate in one simultaneous exchange, and then apply measurements to their local systems. They output values $b',b''\in\{0,1,\perp\}$, where $\perp$ is interpreted as declaring a loss. The players succeed if $b=b'=b'' \neq \perp$, with $b\in\{0,1\}$ determined by measuring a reference maximally entangled with $Q$ in the computational basis if $f(x,y)=0$, and Hadamard basis if $f(x,y)=1$.}
    \label{fig:fBB84}
\end{figure*}

In the next section we will prove a reduction from an $f$-BB84 protocol to an $\SMP^*$ protocol.
To do this, we consider the behaviour of the $f$-BB84 protocol on ideal inputs (the exact BB84 states). 
To constrain the behaviour of the $f$-BB84 protocol on the ideal inputs, given its behaviour on the non-ideal inputs, we need the following continuity statement. 

\begin{lemma}\label{lemma:imperfectpreparation}
    Consider an $f$-BB84 protocol which is $C_\eta$-correct and $L_\eta$-lossy for some set of $\eta$ ideal inputs $\{\rho^{\theta}_x\}$. 
    Then, when given the ideal inputs $V^\theta_x$, the protocol's new correctness and loss parameters $C,L$ satisfy
    \begin{align}
        \left|\left(E_\eta+2sL_\eta\right)-\left(E+2sL\right) \right|\leq \eta
    \end{align}
    where $s=\sin^2(\pi/8)$, and where $E_\eta=1-C_\eta-L_\eta$ and $E=1-C-L$ is the error in the $\eta$-ideal and perfect input cases, respectively. 
\end{lemma}
\begin{proof}\,
    The proof proceeds by using the behaviour of the protocol as a method of distinguishing the possible inputs $\rho^\theta_z$ and $V^\theta_z$. 
    Since the states are close, they must be hard to distinguish, so the behaviour of the protocol cannot be too different. 

    In more detail, we design a distinguishing game as follows. 
    Fix values of $(x,y)$ and $z$.
    The referee computes $\theta=f(x,y)$, then prepares either $\rho^\theta_z$ or $V^\theta_z$, each with probability $1/2$. 
    The referee then gives the resulting state to the player. 
    The player runs the $f$-BB84 protocol for the chosen $(x,y)$, then guesses the state according to the following rules:
    \begin{itemize*}
        \item If the game is correct, guess the state is $V^\theta_z$ with probability $0$. 
        \item If the game is incorrect, guess the state is $V^{\theta}_z$ with probability $1$. 
        \item If the game declares a loss, guess the state is $V^{\theta}_z$ with probability $2s=2\sin^2(\pi/8)$. 
    \end{itemize*}
    Then the probability the player correctly determines the state is
    \begin{align}
        p_{\text{dist}}(V^{\theta}_z, \rho^\theta_z)=&\Pr[V]\Pr[\text{correct}\,|\,V]\Pr[\text{guess}\,V\,|\,\text{correct}] +\nonumber \\
        &\Pr[V]\Pr[\text{incorrect}\,|\,V]\Pr[\text{guess}\,V\,|\,\text{incorrect}]+\nonumber \\
    &\Pr[V]\Pr[\text{loss}\,|\,V]\Pr[\text{guess}\,V\,|\,\text{loss}]+\nonumber\\
    &\Pr[\rho]\Pr[\text{correct}\,|\,\rho]\Pr[\text{guess}\,\rho\,|\,\text{correct}]+\nonumber \\
    &\Pr[\rho]\Pr[\text{incorrect}\,|\,\rho]\Pr[\text{guess}\,\rho\,|\,\text{incorrect}]+\nonumber \\
    &\Pr[\rho]\Pr[\text{loss}\,|\,\rho]\Pr[\text{guess}\,\rho\,|\,\text{loss}].
    \end{align}
    Evaluating this using our probability assignments above, we find that
    \begin{align}
        p_{\text{dist}}(V^{\theta}_z, \rho^\theta_z) = \frac{1}{2}+\frac{1}{2}\left(\left( E+2sL\right)-\left(E_\eta+2sL_\eta\right)\right).
    \end{align}
    Comparing to the Holevo-Helstrom optimal distinguishing probability, we get that
    \begin{align}
        \left( E+2sL\right)-\left( E_\eta+2sL_\eta\right)\leq \eta.
    \end{align}
    We can then repeat the argument, reversing the role of $V$ and $\rho$ in the guessing strategy above to derive
    \begin{align}
        \left( E_\eta+2sL_\eta\right)-\left( E+2sL\right)\leq \eta.
    \end{align}
    Combining the last two claims completes the proof. 
\end{proof}

%%%%%%%%%%%%%%%%%%%%%%%%%%%%%%%%%%%%%%%%%%%%%%%%%%%%%%%%
\subsection{Reduction from \texorpdfstring{$f$}{TEXT}-BB84 to SMP\texorpdfstring{$^*$}{TEXT}}
%%%%%%%%%%%%%%%%%%%%%%%%%%%%%%%%%%%%%%%%%%%%%%%%%%%%%%%%

We prove the following reduction from $f$-BB84 protocols with ideal BB84 state inputs to an SMP$^*$ protocol. 

\begin{theorem}\label{thm:reduction}
    Suppose $P$ is an $f$-BB84 protocol that is $C$-correct and $L$-lossy, where the inputs are distributed according to distribution $\mu$, and where the inputs to the game are exactly the BB84 states.
    Further, suppose the $f$-BB84 protocol uses, in the first round operations, $C_G(f)$ gates drawn from a gate set of size $4$, $C_M(f)$ single-qubit computational basis measurements, and circuits acting on $q$ qubits.
    Then this implies the existence of a $\SMP^*$ protocol with communication cost
    \begin{align}\label{eq:fRgatelowerbound}
        (\log(q)+1)(2C_G(f) + C_M(f)) \geq \SMP^*_{\mu,\epsilon'}(f),
    \end{align}
    where $\SMP^{*}_{\mu,\epsilon'}(f)$ denotes the minimal message size needed to compute $f(x,y)$ in the $\SMP^*$ model with correctness $\epsilon'$ satisfying 
    \begin{align}
        \epsilon' < \frac{1}{2s}\left(E + 2sL\right)
    \end{align}
    where $E=1-C-L$.
\end{theorem}
\begin{proof}\,
We define an $\SMP^*$ protocol as follows. 
We let the referee hold a classical description of the initial resource state shared by Alice and Bob in the $f$-BB84 protocol. 
As well, Alice and Bob share a copy of the resource system from the $f$-BB84 protocol. 
Alice and Bob's strategy will be to send the referee a description of their local operations. 
We consider a decomposition of Alice and Bob's operations into gates and measurements. 
Alice and Bob apply their operations to their shared resource state and the input system. 
As they do so, they keep a record of the gates they apply (which may be computed using mid-circuit measurement outcomes) and their measurement outcomes $m$, then send this to the referee. 
The referee will then compute a classical description of the state $\rho_{\bar{Q}MM'}^{x,y}(m)$, which would result if one end of the maximally entangled state $\Psi^+_{\bar{Q}Q}$ were input to the $f$-BB84 protocol.
For each gate, they specify the gate choice, requiring $2$ bits, and the location of the gate, which requires $2\log q$ bits for a contribution of $(2\log q + 2) C_G(f)$ bits. 
Further, to specify each measurement requires $\log q$ bits to specify where the measurement occurs plus $1$ bit to specify the measurement outcome, for a contribution of $(\log q+1)C_M(f)$. 
The total message size sent by Alice and Bob then is the left hand side of \cref{eq:fRgatelowerbound}. 

We claim that from the description of $\rho_{\bar{Q}MM'}^{x,y}(m)$, the referee can make a guess of $f(x,y)$ which is biased towards the correct value, with a bias that depends on the $(C,L)$ parameters of the $f$-BB84 protocol. 
To understand this, define
\begin{align}
    E^i(\sigma) &= \text{Pr}[\text{at least one player answers, incorrect}\,|\,\text{referee measures in basis} \,\,i] \nonumber \\
    C^i(\sigma) &= \text{Pr}[\text{at least one player answers, correct}\,|\,\text{referee measures in basis} \,\,i] \nonumber \\
    L^i(\sigma) &= \text{Pr}[\text{neither player answers}\,|\,\text{referee measures in basis} \,\,i]
\end{align} 
Note that $C^i+E^i+L^i=1$. 
Note also that these are the probabilities that occur in the $f$-BB84 protocol, if we were to insert the density matrix $\sigma$ into the second round operations.  
Define the averaged (over the choice of basis) quantities,
\begin{align}
    {E}(\sigma) &= \frac{1}{2}\left(E^0+E^1\right),\nonumber \\
    {C}(\sigma) &= \frac{1}{2}\left(C^0+C^1\right),\nonumber \\
    {L}(\sigma) &= \frac{1}{2}\left(L^0+L^1\right).
\end{align}
Note that to complete the $f$-BB84 protocol, Alice and Bob must be able to match the referee's measurement outcomes using the mid-protocol density matrix $\rho_{\bar{Q}MM'}^{x,y}$, with the referee holding $\bar{Q}$. 
This is exactly the lossy MoE game, so \cref{thm:optimal_lossy_bb84_bound} gives
\begin{align}
    C+\frac{L}{\sqrt{2}} \leq \cos^2(\pi/8).
\end{align}
Substituting $C=1-E-L$, defining $s=\sin^2(\pi/8)$, and simplifying, we find
\begin{align}
    2s \leq (E^0+2sL^0) + (E^1+2sL^1). 
\end{align}
Define $Z^i(\sigma)\equiv E^i(\sigma)+2sL^i(\sigma)$, so that
\begin{align}
    \boxed{2s \leq Z^0+Z^1}
\end{align}
is the MoE bound. 

The variable $Z^i$ is small for $f(x,y)=i$ when the protocol is highly correct, in the sense that both the error and loss are low. 
The value of $Z^i$ depends on the density matrix $\rho$, and is computed from the probabilities resulting from running Alice and Bob's strategies using the density matrix $\rho$.
We can also consider $Z^i$ for the post-selected density matrix $\rho_{\bar{Q}MM'}^{x,y}(m)$, which we abbreviate as $Z_{m,x,y}^i\equiv Z^i(\rho^{x,y}(m))$. 
This is computed from the probabilities that occur using the measurement strategy defined by the $f$-BB84 protocol, but now replacing $\rho^{x,y}$ with $\rho^{x,y}(m)$. 
We also define
\begin{align}
    Z_{\opt}^i(\sigma) = \min_{\text{strategies}}Z^i(\sigma)
\end{align}
where the minimization is over all allowed measurement strategies, which act separately on $M$ and $M'$. 
Notice that while the referee in the $\SMP^*$ protocol may not be able to compute $Z^i(\rho^{x,y}(m))$ (since it depends on Alice and Bob's choice of measurements, which depends on $x,y$) they can compute $Z_{\opt}^i(\rho^{x,y}(m))$, since this only depends on the density matrix $\rho^{x,y}(m)$, which they know. 
We can also observe that
\begin{align}
    Z_{\opt}^0(\rho^{x,y}(m))+ Z_{\opt}^1(\rho^{x,y}(m)) \geq 2s
\end{align}
by the MoE bound, which applies for all strategies, and that
\begin{align}
    Z^{i}_{\opt}(\rho^{x,y}(m))\leq Z_{m,x,y}^i
\end{align}
since the actual strategy used in the $f$-BB84 protocol may be worse than the optimal one. 

We take the referee's strategy in the SMP protocol to be as follows. 
Given a description of $\rho^{x,y}(m)$, they compute $Z^{0}_{\opt}(\rho^{x,y}(m))$ and $Z^{1}_{\opt}(\rho^{x,y}(m))$. 
Then, if $Z^{i}_{\opt}(\rho^{x,y}(m))$ is the smaller of the two, they bias their response towards $i$. 
This is sensible, as this is saying that $\rho^{x,y}(m)$ is better for guessing the outcome in basis $i$ than $\neg i$, so it seems likely to be sampled from a $\rho^{x,y}$ that works well for basis $i$, which occurs when $f(x,y)=i$. 
More specifically, they compute $D(m,x,y)=Z^{1}_{\opt}(\rho^{x,y}(m))-Z^{0}_{\opt}(\rho^{x,y}(m))$ and output $0$ with probability
\begin{align}\label{eq:p0}
    p_0(m,x,y)=\frac{1}{2}+\frac{D(m,x,y)}{4s}.
\end{align}
Note that $|D|\leq 2s$, because $Z^{j}_{\opt}(\rho^{x,y}(m))\leq 2s$ (using the strategy that always outputs $\perp$, so has $E=0,L=1$), so $0\leq p_0 \leq 1$ is a probability.

Now say that the actual basis is $0$. 
Then $Z^{0}_{\opt}(\rho^{x,y}(m))\leq Z_{m,x,y}^0$, so 
\begin{align}
    D(m,x,y) &= Z^{1}_{\opt}(\rho^{x,y}(m))-Z^{0}_{\opt}(\rho^{x,y}(m)) \nonumber \\
    &\geq 2s - 2Z^{0}_{\opt}(\rho^{x,y}(m)) \nonumber \\
    &\geq 2s - 2Z_{m,x,y}^0
\end{align}
where we used $Z^0_{\opt}+Z^1_{\opt}\geq 2s$ (MoE) in the first inequality, and $Z^{0}_{\opt}(\rho^{x,y}(m))\leq Z_{m,x,y}^0$ in the second inequality. 
But then on average over $m$, this is
\begin{align}
    \langle D(m,x,y) \rangle \geq 2s-2\langle Z_{m,x,y}^0\rangle  
\end{align}
so we need to bound $\langle Z_{m,x,y}^0\rangle$. 
This is
\begin{align}
    \langle Z_{m,x,y}^0 \rangle = \langle E_{m,x,y}^0\rangle  + 2s\langle L_{m,x,y}^0 \rangle
\end{align}
where $E_{m,x,y}^0$ is the probability at least one of the players respond and the response is incorrect, and $L_{m,x,y}^0$ is the probability both players don't respond. 
We have that
\begin{align}
    \langle E_{m,x,y}^0 \rangle &= E^{0}_{x,y} \nonumber \\
    \langle L^0_{m,x,y}\rangle &= L^0_{x,y}
\end{align}
so that
\begin{align}
    \langle Z_{m,x,y}^0 \rangle = E^0_{x,y}+2sL^0_{x,y}.
\end{align}
This gives a lower bound on $\langle D(m,x,y)\rangle $ of
\begin{align}
    \langle D(m,x,y)\rangle \geq 2s-2(E^0_{x,y}+2sL^0_{x,y}).
\end{align}
Returning to \cref{eq:p0}, we have that $p_0(x,y)=\langle p_0(m,x,y)\rangle$, when $(x,y)\in f^{-1}(0)$, is at least
\begin{align}
    p_{0}(x,y) \geq \frac{1}{2} + \frac{1}{2}\left(1 - \frac{1}{s}\left({E^0_{x,y}} +2sL_{x,y}^0\right)\right).
\end{align}
Similarly, we find that when $(x,y)\in f^{-1}(1)$, 
\begin{align}
    p_{1}(x,y) \geq \frac{1}{2} + \frac{1}{2}\left(1 - \frac{1}{s}\left({E^1_{x,y}} +2sL_{x,y}^1\right)\right).
\end{align}
So then the winning probability, averaged over inputs $(x,y)$, is
\begin{align}
    p_{win}&=\sum_{x,y}\mu(x,y) \,p_{f(x,y)}(x,y),  \nonumber \\
    &\geq \frac{1}{2} + \frac{1}{2}\left(1 - \frac{1}{s}\left(E +2sL\right)\right).
\end{align}
Then using $p_{win}=1-\epsilon'$, we obtain the claimed bound on $\epsilon'$.
\end{proof} 

\begin{remark}
    \textbf{Slow quantum information:} Notice that the proof of \cref{thm:reduction} does not use that the input state is of the specific form $\Psi^+_{\bar{Q}Q}\otimes \Psi_{AB}$, which is what occurs in the general protocol in \cref{sec:defbb84}. Instead, the proof only needs that the state input to Alice and Bob's first round operations is fixed (independent of $(x,y)$). 
    This means the reduction also applies when the input system $Q$ is given early, which allows Alice and Bob to prepare any state of the form $\rho_{\bar{Q}AB}=\mathcal{P}_{Q\rightarrow AB}(\Psi^+_{\bar{Q}Q})$ and distribute $A$, $B$, before receiving $x$ and $y$.
    In practice, this is desirable in implementations where the honest player sends quantum information over fibre optics, since then the quantum information travels slower than the speed of light. 
    This means it should be sent early, necessitating the above remarks. 
\end{remark}

Next, we combine the above reduction with \cref{lemma:imperfectpreparation} to bound an $f$-BB84 protocol with imperfect inputs. 

\begin{theorem}\label{thm:reduction_imperfect}
    Suppose $P$ is an $f$-BB84 protocol that is $C$-correct and $L$-lossy, where the inputs are distributed according to distribution $\mu$, and where the inputs to the game are $\eta$-close in trace distance to the ideal states.
    Further, suppose the $f$-BB84 protocol uses, in the first round operations, $C_G(f)$ gates drawn from a gate set of size $4$, $C_M(f)$ single-qubit computational basis measurements, and circuits acting on $q$ qubits.
    Then this implies the existence of a $\SMP^*$ protocol with communication cost
    \begin{align}
        (\log(q)+1)(2C_G(f) + C_M(f)) \geq \SMP^*_{\mu,\epsilon'}(f),
    \end{align}
    where $\SMP^{*}_{\mu,\epsilon'}(f)$ denotes the minimal message size needed to compute $f(x,y)$ in the $\SMP^*$ model with correctness $\epsilon'$ satisfying 
    \begin{align}
        \epsilon' < \frac{1}{2}\left(\frac{\epsilon(1-L)}{s} + 2L+\frac{\eta}{s}\right).
    \end{align}
    where $\epsilon=\Pr[\text{incorrect response}\,|\,\text{at least one player responds}\,]$.
\end{theorem}
\begin{proof}\,
    By \cref{lemma:imperfectpreparation}, we have that if we take the given $f$-BB84 protocol (with correctness $C$ and loss $L$) and feed it the ideal input states, we have correctness $C'$ and loss $L'$ satisfying
    \begin{align}
        \left| (E+2sL)-(E'+2sL')\right|\leq \eta
    \end{align}
    Applying \cref{thm:reduction} then, we obtain an SMP protocol with
    \begin{align}
        \epsilon' < \frac{1}{2s}\left(E'+2sL'\right)
    \end{align}
    so then, in terms of the parameters $C,L,E=1-C-L$ describing the protocol given non-ideal inputs,
    \begin{align}
        \epsilon' < \frac{1}{2s}\left(E'+2sL'\right) \leq \frac{1}{2s}\left(E+2sL+\eta \right).
    \end{align}
    Finally we can use that
    \begin{align}
        E&=\Pr[\text{incorrect response}] \nonumber \\
        &=\Pr[\text{incorrect response}\,|\,\text{at least one player responds}]\Pr[\text{at least one player responds}] \nonumber \\
        &=\epsilon(1-L)
    \end{align}
    which yields the final expression.
\end{proof}

Finally we are ready to apply this to the $f$-BB84 protocol with $f$ chosen to be the inner product function. 

\begin{corollary}\label{corrolary:IPapplication}
    Consider any $f$-BB84 protocol with $f$ chosen to be the inner product function, 
    \begin{align}
        \IPfunc(x,y) = \sum_{i=1}^n x_i y_i \,\, \text{mod}\,\, 2\,.
    \end{align}
    Assume the protocol has
    \begin{itemize*}
        \item $L:=\Pr[\text{both players declare a loss}]$
        \item $1-\epsilon:=\Pr[\text{both players respond correctly}\,|\, \text{at least one player responds}]$
    \end{itemize*}
    Further, suppose the inputs are distributed according to the uniform distribution, that the inputs are $\eta$-ideal, and that the first round operations involve $C_G$ two qubit gates drawn from a set of size $4$, $C_M$ single-qubit measurements, and act on at most $q$ qubits (shared between the left and right). 
    Then, the first round operations in any such $f$-BB84 protocol must satisfy
    \begin{align}
        (\log(q)+1)(2C_G + C_M) \geq n + 2\log(1-2\epsilon')
    \end{align}
    where $\epsilon'=\frac{1}{2}\left(\frac{\epsilon(1-L)}{s} + 2L+\frac{\eta}{s}\right)$, whenever $\left(\frac{\epsilon(1-L)}{s} + 2L+\frac{\eta}{s}\right)<1$, corresponding to the security region shown in \cref{fig:securityregion}.
\end{corollary}

\noindent \textbf{Optimal security region:} We now show that when $\eta=0$, our security region, defined by,
\begin{align}
    \frac{\epsilon(1-L)}{s} + 2L < 1,
\end{align}
is as large as possible, in that there is an attack using $O(1)$ quantum gates which can be used to attack the scheme everywhere outside that region. 
The strategy is the same as was used to show optimality of the bound \cref{thm:optimal_lossy_bb84_bound} on lossy MoE games.
Concretely, with probability $p$ Alice and Bob measure in the Breidbart basis, then both return this measurement outcome as their result. 
This can be checked to fail with probability $s=\sin^2(\pi/8)$, and never returns a loss. 
With probability $1-p$, Alice and Bob follow an alternative strategy: Alice measures in the Hadamard or computational bases each with probability $1/2$, sends her outcome to Bob, and they both output the resulting value if they learn they've guessed in the correct basis. 
Otherwise, they declare a loss. 
This strategy leads to
\begin{align}
    E &=\epsilon(1-L) = ps, \nonumber \\
    L &= \frac{1-p}{2}.
\end{align}
These values saturate the bound $\frac{\epsilon(1-L)}{s}+2L\leq 1$, as claimed.

\vspace{0.2cm}
\noindent \textbf{Extension to committed $f$-BB84:} In \cite{allerstorfer2023making} the authors define a variant of the $f$-BB84 scheme involving a commitment step, which is secure even given arbitrarily high transmission loss. 
Our gate lower bound can also be applied to the committed protocol. 
We briefly outline how. 
In the protocol with commitment, the quantum input is given early, allowing Alice and Bob to apply an arbitrary map to it and prepare a bipartite state $\rho_{AB}$, with $A$ held by Alice and $B$ held by Bob.
Due to the relativistic structure of the protocol with commitment, the decision to commit or not must be made by each player independently and in the first round operations.
Thus the attackers in the committed protocol can be modelled as first, upon being given input $x$ (on the left) or $y$ (on the right), locally making a two outcome measurement (which can depend on $x$ or $y$) of their quantum system $A$ or $B$, to determine if they should commit or declare a loss.
Since in the honest protocol the outputs on the left and right will either both commit or not commit, we assume the attackers mimic this behaviour, and so their outputs from this measurement must agree. 
The work \cite{allerstorfer2023making} then shows\footnote{See Lemma 4.8 of the arXiv version, which is a robust version of this claim.} that this means there is a single state $\rho^*_{AB}$ which is close to all of the states $\rho^{x,y}_{AB}$ which result from obtaining inputs $(x,y)$ and post-selecting on both players committing.
This means conditioning on a commitment does not help the players: they could have done away with the commitment and just prepared the state $\rho^*$ initially. 
Since this is a viable strategy in the protocol without a commitment, the players must (after the commitment measurements) then apply a strategy viable for the protocol without commitment. 
In other words, their protocol after the commitment step must use a number of gates and measurements compatible with our bound.

%%%%%%%%%%%%%%%%%%%%%%%%%%%%%%%%%%%%%%%%%%%%%%%%%%%%%%%%%%%%%%%
\section{Discussion}
%%%%%%%%%%%%%%%%%%%%%%%%%%%%%%%%%%%%%%%%%%%%%%%%%%%%%%%%%%%%%%%

In this work we extend earlier gate lower bounds on the $f$-BB84 QPV scheme to allow for loss tolerance, of up to $50\%$ at zero error, and improve the error tolerance at zero (declared) loss, of up to $\epsilon=\sin^2(\pi/8)$. 
Without modifying the $f$-BB84 scheme, we achieve the largest security region possible, at least in the case of perfect BB84 states being input to the game. 

We also extend our analysis to consider imperfect input states, characterized in terms of a bound on the trace distance, which recall we labelled $\eta$.  
We have not shown that our bound is tight in terms of its behaviour with $\eta$. 
Perhaps more importantly, it is not clear to us that the trace distance is the most natural way of characterizing the imperfections in the input states. 
For instance, we might instead want to devise an experimentally feasible test on the input states which, when passed, ensures the scheme is secure within some corresponding parameter regime.

While we work in the context of qubit systems, with experimental applications in mind it would be natural to extend our results to the context of coherent states. 
This introduces new complications, and we leave this to future work. 
Another point of practical concern is to consider finite statistical tests. 
For instance, while we show that succeeding in the $f$-BB84 game with a given success and loss probability requires linear gates, in practice we have a finite set of runs of the game. 
We then need to consider the probability with which the actual success and loss probabilities are within the security region.

In our gate lower bounds we assume a simple gate set of size 4, acting on at most 2 qubits. 
In practice, many quantum computing architectures allow continuous gate sets. 
Only a certain precision is possible in applying these gates, so that we can use the Solovay-Kitaev theorem to replace this continuous gate set with a finite one, and then apply our bound. 
More broadly, it may be best to view our bound as being directly on the descriptive complexity of the quantum operations being performed, and make our assumption on the attacker be that they cannot apply quantum operations with too large of a descriptive complexity. 

\vspace{0.2cm}
\noindent \textbf{Acknowledgements:} A. M.
acknowledges the support of the Natural Sciences and
Engineering Research Council of Canada (NSERC); this
work was supported by an NSERC-UKRI Alliance Grant
No. (ALLRP 597823-24). Research at Perimeter Institute is supported in part by the Government of Canada through the Department of Innovation, Science and Economic Development and by the Province of Ontario through the Ministry of Colleges and Universities. PVL is supported by France 2030 under the French National Research Agency award number ANR-22-PETQ-0007. 
Claude Opus 4 (Anthropic) was used to suggest improvements to the proof of \cref{thm:reduction}, which resulted in a substantial increase in the size of the security region compared to our initial strategy. 
ChatGPT 5.6 (OpenAI) was used to discuss proof ideas related to Theorem \ref{thm:optimal_lossy_bb84_bound}.

\bibliographystyle{unsrtnat}
\bibliography{biblio}

\end{document}